\documentclass[11pt, reqno, oneside]{article}
\usepackage{geometry} 
\usepackage[affil-it]{authblk}

\usepackage{amsmath, amsthm, amsfonts, amssymb, amsbsy, bigstrut, graphicx, enumerate,  upref, longtable, comment, booktabs, array, caption, subcaption}

\usepackage{times}

\usepackage[breaklinks]{hyperref}

\usepackage[numbers, sort&compress]{natbib}

\newtheorem{thm}{Theorem}[section]

\theoremstyle{definition}

\newtheorem{ex}[thm]{Example}

\newcommand{\bbx}{\mathbf{X}}

\newcommand{\bbz}{\mathbf{Z}}

\newcommand{\ee}{\mathbb{E}}

\newcommand{\mx}{\mathcal{X}}

\newcommand{\my}{\mathcal{Y}}

\newcommand{\pp}{\mathbb{P}}

\newcommand{\rr}{\mathbb{R}}

\newcommand{\var}{\mathrm{Var}}

\numberwithin{equation}{section}

\usepackage[utf8]{inputenc}
\usepackage{amsmath}
\usepackage{amsfonts}
\usepackage{bbm}
\usepackage{mathtools}
\usepackage{tikz}
\usetikzlibrary{decorations.pathreplacing}
\usepackage{amsthm}
\usepackage[english]{babel}
\usepackage{fancyhdr}
\usepackage{amssymb}
\usepackage{enumitem}
\usepackage{qtree}
\usepackage{mathrsfs}

\newcommand{\R}{\mathbb{R}}

\renewcommand{\hat}{\widehat}

\newcommand{\mz}{\mathcal{Z}}

\begin{document}
\title{A survey of some recent developments in measures of association}
\author{Sourav Chatterjee\thanks{Mailing address: Department of Statistics, Stanford University, 390 Jane Stanford Way, Stanford, CA 94305, USA. Email: \href{mailto:souravc@stanford.edu}{\tt souravc@stanford.edu}. The author was partially supported by NSF grants DMS-2113242 and DMS-2153654. The author thanks Nabarun Deb, Fang Han  and Bodhisattva Sen for helpful comments on a preliminary draft.
}}
\affil{Stanford University}

\maketitle

\begin{center}
{\it \small In honor of friend and teacher Prof.~Rajeeva L.~Karandikar on the occasion of his 65$^{th}$ birthday.}
\end{center}

\begin{abstract}
This paper surveys some recent developments in measures of association related to a new coefficient of correlation introduced by the author. A straightforward extension of this coefficient to standard Borel spaces (which includes all Polish spaces), overlooked in the literature so far, is proposed at the end of the survey.  
\end{abstract}

\section{Introduction}
Measuring associations between variables is one of the central goals of data analysis.  Arguably, the three most popular classical measures of association are Pearson's correlation coefficient, Spearman's $\rho$, and Kendall's $\tau$. Although these coefficients are powerful for detecting monotonic associations, a practical problem is that  they are not effective for detecting associations that are not monotonic. There have been many proposals to address this deficiency of the classical coefficients~\cite{josseholmes16}, such as the maximal correlation coefficient~\cite{hirschfeld35, gebelein41, renyi59, breimanfriedman85}, various coefficients based on joint cumulative distribution functions and ranks~\cite{gamboaetal18, drtonetal20, hanetal17, bergsmadassios14, nandyetal16, weihsetal16, weihsetal18, yanagimoto70, csorgo85, purisen71, hoeffding48, blumetal61, romano88, rosenblatt75, debsen21, wangetal17}, kernel-based methods~\cite{pfisteretal18, grettonetal05, grettonetal07, sensen14, zhangetal18}, information theoretic coefficients~\cite{kraskovetal04, linfoot57, reshefetal11}, coefficients based on copulas~\cite{detteetal13, lopez-pazetal13, sklar59, schweizerwolff81, zhang19}, and coefficients based on  pairwise distances~\cite{szekelyetal07, szekelyrizzo09, helleretal13, friedmanrafsky83, lyons13}.  

This survey is about some recent developments in this area, beginning with a new coefficient of correlation proposed by the author in the paper~\cite{chatterjee21a}. This coefficient has the following desirable features: (a) It has a simple expression, like the classical coefficients. (b) It is a consistent estimator of a measure of dependence which is $0$ if and only if the variables are independent and  $1$ if and only if one is a measurable function of the other. (c) It has a simple asymptotic theory under the hypothesis of independence.

The new coefficient is defined as follows. Let $(X,Y)$ be a pair of random variables defined on the same probability space, where $Y$ is not  a constant.  Let $(X_1,Y_1),\ldots,(X_n,Y_n)$ be i.i.d.~pairs of random variables with the same law as $(X,Y)$, where $n\ge 2$. Rearrange the data as $(X_{(1)},Y_{(1)}),\ldots,(X_{(n)}, Y_{(n)})$ such that $X_{(1)}\le \cdots \le X_{(n)}$. (Note that $Y_{(i)}$ is just the $Y$-value `paired with' $X_{(i)}$ in the original data, and not the $i^{th}$ order statistic of the $Y$'s.)  If there are ties among the $X_i$'s, then choose an increasing rearrangement as above by breaking ties uniformly at random. Let $r_i$ be the number of $j$ such that $Y_{(j)}\le Y_{(i)}$, and let $l_i$ to be the number of $j$ such that $Y_{(j)}\ge Y_{(i)}$. Then define
\begin{align}\label{xindef}
\xi_n(X,Y) := 1-\frac{n\sum_{i=1}^{n-1}|r_{i+1}-r_i|}{2\sum_{i=1}^n l_i(n-l_i)}. 
\end{align}
This is the correlation coefficient proposed in \cite{chatterjee21a}. When there are no ties among the $Y_i$'s, $l_1,\ldots,l_n$ is just a permutation of $1,\ldots,n$, and so the denominator in the above expression is just $n(n^2-1)/3$. The following theorem is the main consistency result for $\xi_n$. 


\begin{thm}[\cite{chatterjee21a}]\label{mainthm}
If $Y$ is not almost surely a constant, then as $n\to\infty$, $\xi_n(X,Y)$ converges almost surely to the deterministic limit
\begin{equation}\label{xidef}
\xi(X,Y) := \frac{\int \var(\ee(1_{\{Y\ge t\}}|X)) d\mu(t)}{\int\var(1_{\{Y\ge t\}}) d\mu(t)},
\end{equation}
where $\mu$ is the law of $Y$. This limit belongs to the interval $[0,1]$. It is  $0$ if and only if $X$ and $Y$ are independent, and it is $1$ if and only if there is a measurable function $f:\rr\to\rr$ such that $Y=f(X)$ almost surely. 
\end{thm}

 The limiting value $\xi(X,Y)$ appeared in the literature prior to \cite{chatterjee21a}, in a  paper of~\citet*{detteetal13} (see also~\cite{gamboaetal18, kongetal19}). The paper \cite{detteetal13} gave a copula-based estimator for $\xi(X,Y)$ when $X$ and $Y$ are continuous, that is consistent under smoothness assumptions on the copula and is computable in time $n^{5/3}$ for an optimal choice of tuning parameters. 
 
Note that neither $\xi_n$ nor the limiting value $\xi$ is symmetric in $X$ and $Y$. A symmetrized version of $\xi_n$ can be constructed by taking the maximum of $\xi_n(X,Y)$ and $\xi_n(Y,X)$.

\section{Why does it work?}
The complete proof of Theorem \ref{mainthm} is available in the supplementary materials of \cite{chatterjee21a}, and also in the arXiv version of the paper. It is not too hard to see why $\xi(X,Y)$ has the properties listed in Theorem \ref{mainthm}, although filling in the details takes some work. The proof of convergence of $\xi_n(X,Y)$ to $\xi(X,Y)$ is less obvious. The following is a very rough sketch of the proof, reproduced from a similar discussion in \cite{chatterjee21a}. 

For simplicity, consider only the case of continuous $X$ and $Y$, where the denominator in~\eqref{xidef} is simply $n(n^2-1)/3$.  First, note that by the Glivenko--Cantelli theorem, $r_i/n \approx F(Y_{(i)})$, where $F$ is the cumulative distribution function of $Y$. Thus, 
\begin{align}\label{xiapprox}
\xi_n(X,Y) \approx 1 - \frac{3}{n} \sum_{i=1}^n |F(Y_i)- F(Y_{N(i)})|,
\end{align}
where $N(i)$ is the unique index $j$ such that $X_j$ is immediately to the right of $X_i$ if we arrange the $X$'s in increasing order. If $X_i$ is the rightmost value, define $N(i)$ arbitrarily; it does not matter since the contribution of a single term in the above sum is of order $1/n$. Next,  observe that for any $x,y\in \rr$, 
\begin{align}\label{fform}
|F(x)- F(y)| &= \int (1_{\{t\le x\}} - 1_{\{t\le y\}})^2 d\mu(t),
\end{align}
where $\mu$ is the law of $Y$. This is true because the integrand is $1$ between $x$ and $y$ and $0$ outside.

Suppose that we condition on $X_1,\ldots,X_n$. Since $X_i$ is likely to be very close to $X_{N(i)}$, the random variables $Y_i$ and $Y_{N(i)}$ are likely to be approximately i.i.d.~after this conditioning. This leads to the approximation
\begin{align*}
\ee[(1_{\{t\le Y_i\}} - 1_{\{t\le Y_{N(i)}\}})^2|X_1,\ldots,X_n] &\approx 2\var(1_{\{t\le Y_i\}}|X_1,\ldots,X_n) \\
&= 2\var(1_{\{t\le Y_i\}}|X_i). 
\end{align*}
This gives
\begin{align*}
\ee(1_{\{t\le Y_i\}} - 1_{\{t\le Y_{N(i)}\}})^2 &\approx  2\ee[\var(1_{\{t\le Y\}}|X)]\\
&= 2\var(1_{\{t\le Y\}}) - 2 \var(\ee(1_{\{t\le Y\}}|X)). 
\end{align*}
Combining this with \eqref{fform}, we get
\begin{align*}
\ee|F(Y_i)-F(Y_{N(i)})| &\approx \int 2[\var(1_{\{t\le Y\}}) -\var(\ee(1_{\{t\le Y\}}|X))]d\mu(t). 
\end{align*}
But note that $\var(1_{\{t\le Y\}}) = F(t)(1-F(t))$, and $F(Y)\sim \textup{Uniform}[0,1]$. Thus, 
\[
\int \var(1_{\{t\le Y\}}) d\mu(t) = \int F(t)(1-F(t)) d\mu(t) = \int_0^1 x(1-x)dx = \frac{1}{6}.
\] 
Therefore by \eqref{xiapprox},
\begin{align*}
\ee(\xi_n(X,Y)) \approx 6 \int \var(\ee(1_{\{t\le Y\}}|X))d\mu(t) = \xi(X,Y), 
\end{align*}
where the last identity holds because $\int \var(1_{\{t\le Y\}}) d\mu(t)  = 1/6$, as shown above. This establishes the convergence of $\ee(\xi_n(X,Y))$ to $\xi(X,Y)$. Concentration inequalities are then used to show that $\xi_n(X,Y)-\ee(\xi_n(X,Y))\to 0$ almost surely.

\section{Asymptotic distribution}\label{asympsec}
Let $X$, $Y$ and $\xi_n$ be as in the previous section. For each $t\in \rr$, let $F(t):= \pp(Y\le t)$ and $G(t):=\pp(Y\ge t)$. Let $\phi(y,y') := \min\{F(y), F(y')\}$. Define 
\begin{equation}\label{tauform}
\tau^2 = \frac{\ee \phi(Y_1,Y_2)^2 - 2\ee(\phi(Y_1,Y_2)\phi(Y_1,Y_3)) + (\ee\phi(Y_1,Y_2))^2}{(\ee G(Y)(1-G(Y)))^2},
\end{equation}
where $Y_1,Y_2,Y_3$ are independent copies of $Y$. The following theorem gives the limiting distribution of $\xi_n$ under the null hypothesis that $X$ and $Y$ are independent.
\begin{thm}[\cite{chatterjee21a}]\label{cltthm}
Suppose that $X$ and $Y$ are independent. Then   $\sqrt{n}\xi_n(X,Y)$ converges to $N(0, \tau^2)$ in distribution as $n\to\infty$, where $\tau^2$ is given by the formula~\eqref{tauform} stated above. The number $\tau^2$ is strictly positive if $Y$ is not a constant, and equals $2/5$ if $Y$ is continuous. 
\end{thm}

The reason why $\tau^2$ does not depend on the law of $Y$ if $Y$ is continuous is that in this case $F(Y)$ and $G(Y)$ are Uniform$[0,1]$ random variables, which implies that the expectations in \eqref{tauform} do not depend on the law of $Y$. If $Y$ is not continuous, then $\tau^2$ may depend on the law of $Y$. For example, it is not hard to show that if $Y$ is a Bernoulli$(1/2)$ random variable, then $\tau^2=1$. Fortunately, if $Y$ is not continuous, there is a simple way to estimate $\tau^2$ from the data using the estimator
\[
\hat{\tau}^2_n = \frac{a_n-2b_n+c_n^2}{d_n^2},
\]
where  $a_n$, $b_n$, $c_n$ and $d_n$ are defined as follows. For each $i$, let 
\begin{equation}\label{rldef}
R(i) := \#\{j: Y_j\le Y_i\}, \ \ \ L(i) := \#\{j: Y_j\ge Y_i\}. 
\end{equation}
(Note that $R(i)$ and $L(i)$ are different than $r_i$ and $l_i$ defined earlier.) Let $u_1\le u_2\le\cdots \le u_n$ be an increasing rearrangement of $R(1),\ldots, R(n)$. Let $v_i := \sum_{j=1}^i u_j$ for $i=1,\ldots,n$.  Define
\begin{align*}
&a_n := \frac{1}{n^4}\sum_{i=1}^n (2n-2i+1) u_i^2, \ \ \ b_n := \frac{1}{n^5}\sum_{i=1}^n (v_i + (n-i)u_i)^2, \\
&c_n := \frac{1}{n^3}\sum_{i=1}^n (2n-2i+1)u_i, \ \ \  d_n := \frac{1}{n^3}\sum_{i=1}^n L(i)(n-L(i)). 
\end{align*}
Then the following holds.
\begin{thm}[\cite{chatterjee21a}]\label{estthm}
The estimator $\hat{\tau}_n^2$ can be computed in time $O(n\log n)$, and converges to $\tau^2$ almost surely as $n\to\infty$. 
\end{thm}
The question of proving a central limit theorem for $\xi_n$ in the absence of independence is much more difficult than the independent case. This was left as an open question in \cite{chatterjee21a} and recently resolved in complete generality by \citet{linhan22}, following earlier proofs of \citet*{debetal20} and \citet*{shietal21} under additional assumptions. \citet{linhan22} also give a consistent estimator of the asymptotic variance of $\xi_n$ in the absence of independence, solving another question that was left open in \cite{chatterjee21a}. A central limit theorem for the symmetrized version of $\xi_n$ (defined as the maximum of $\xi_n(X,Y)$ and $\xi_n(Y,X)$) under the hypothesis of independence was proved by~\citet{zhang22}.

\section{Power for testing independence}
A deficiency of $\xi_n$, as already pointed out in \cite{chatterjee21a} through simulated examples, is that it has low power for testing independence against `standard' alternatives, such as linear or monotone associations. This was theoretically confirmed by~\citet*{shietal22}, where it was shown that the test of independence using $\xi_n$ is rate-suboptimal against a family of local alternatives, whereas three other nonparametric tests of independence proposed in \cite{hoeffding48, blumetal61, bergsmadassios14, yanagimoto70} are rate-optimal. Like $\xi_n$, the three competing test statistics considered in \cite{shietal22} are also computable in $O(n\log n)$ time. Similar results were obtained for a different type of competing test statistic by  \citet{caobickel20}. 

A more detailed analysis of the power properties of $\xi_n$ was carried out by \citet*{auddyetal21}, where the asymptotic distribution of $\xi_n$ under any changing sequence of alternatives converging to the null hypothesis of independence was computed. This analysis yielded exact detection thresholds and limiting power under natural alternatives converging to the null, such as mixture models, rotation models and noisy nonparametric regression. The detection boundary lies at distance $n^{-1/4}$ from the null, instead of the more standard $n^{-1/2}$. This is similar to the power properties of other `graph-based' statistics for testing independence, such as the Friedman--Rafsky statistic~\cite{friedmanrafsky83, bhattacharya19}.

A proposal for `boosting' the power of $\xi_n$ for testing independence, by incorporating multiple nearby ranks instead of only the nearest ones, was recently proposed by \citet{linhan21}. This modified estimator was shown to attain near-optimal rates of power against certain classes of alternative hypotheses.

The conceptual reason behind the absence of local power of statistics such as $\xi_n$ was explained by \citet{bickel22}. An interesting question that remains unexplained is the following. It is seen in simulations that although $\xi_n$ has low power for testing independence against standard alternatives such as linear and monotone, it becomes more powerful as the signal starts to get more and more oscillatory \cite{chatterjee21a}. This gives an advantage over other coefficients in applications where oscillatory signals arise naturally~\cite{chen20, sadeghi22}, and suggests that $\xi_n$ may be efficient for certain kinds of local alternatives. No result of this sort has yet been proven.

\section{Multivariate extensions}\label{multsec}
Many methods have been proposed for testing independence nonparametrically in the multivariate setting. This includes classical tests~\cite{purietal70, grettonetal05a, szekelyetal07} as well as a flurry of recent ones proposed in the last ten years~\cite{helleretal12, hellerheller16, helleretal13, zhuetal17, weihsetal18, kimetal20, debsen21, shietal22a, berrettetal21, shietal22b, debetal20}. 

Most of these papers are concerned only with testing independence, and not with measuring the strength of dependence as measured by a correlation coefficient such as $\xi_n$. Unfortunately, the $\xi_n$ coefficient and many other popular univariate coefficients do not readily generalize to the multivariate setting because they are based on ranks. Statisticians have started taking a new look at this old problem in recent years by considering a multivariate notion of rank defined using optimal transport. Roughly speaking, the idea is as follows. Let $\nu$ be a `reference measure' in $\R^d$, akin to the uniform distribution on $[0,1]$ in $\R$. Given any probability measure $\mu$ in $\R^d$, let $F^\mu:\R^d \to \R^d$ be the map that `optimally transports' $\mu$ to $\nu$ --- that is, if $X\sim \mu$ then $F^\mu(X) \sim \nu$, and $F^\mu$ minimizes $\ee\|X - F(X)\|^2$ among all such $F$. By a theorem of \citet{mccann95}, such a map exists and is unique if $\mu$ and $\nu$ are both absolutely continuous with respect to Lebesgue measure. For example, when $d=1$, $F^\mu$ is just the cumulative distribution function of $X$, which transforms $X$ into a Uniform$[0,1]$ random variable. For properties of this map, see, e.g.,~\citet{figalli18} and \citet{hallinetal21}. 

The above idea suggests a natural definition of multivariate rank. If $X_1,\ldots,X_n$ are i.i.d.\ samples from $\mu$, one can try to estimate $F^\mu$ using this data. Let $F_n^\mu$ be an estimate. Then $F_n^\mu(X_i)$ can act as a `multivariate rank' of $X_i$ among $X_1,\ldots, X_n$, divided by $n$. Since $F^\mu(X_i)\sim \nu$, we can then assume that $F_n^\mu(X_i)$ is approximately distributed according to $\nu$, and then try to test for independence of random vectors using a test for independence that works when the marginal distributions are both $\nu$. This idea has been made precise in a number of recent works in the statistics literature, such as \citet{chernozhukovetal17},   \citet*{debetal20}, \citet{debsen21},  \citet{hallinetal21}, \citet{manoleetal21}, \citet*{shietal22a}, \citet{ghosalsen22}, \citet{mordantsegers22} and \citet{shietal22b, shietal21a}. For a survey and further discussions, see \citet{han21}.

A direct generalization of $\xi_n$ to higher dimensional spaces has not been proposed so far, although the variant proposed in \citet{debetal20} satisfies the same properties as $\xi_n$  provided that the space on which $Y$ takes values admits a nonnegative definite kernel and the  space on which $X$ takes values has a metric. This covers most spaces that arise in practice. There are a couple of other generalizations, proposed by \citet{azadkiachatterjee21} and \citet{gamboaetal22}, on measuring the dependence between a univariate random variable $Y$ and a random vector $\bbx$. The coefficient proposed in \cite{azadkiachatterjee21} (discussed in detail in Section \ref{condsec}) is based on a generalization of the ideas behind the construction of $\xi_n$. The coefficient  proposed in~\cite{gamboaetal22} combines the construction of $\xi_n$ with   ideas from the theory of Sobol indices. A new contribution of the present paper is a simple generalization of $\xi_n$ to standard Borel spaces, which has been overlooked in the literature until now. This is presented in Section \ref{newsec}. 

\section{Measuring conditional dependence}\label{condsec}
The problem of measuring conditional dependence has received less attention than the problem of measuring unconditional dependence, partly because it is a more difficult task. Non-parametric conditional independence can be tested for discrete data using the classical Cochran--Mantel--Haenszel test \cite{cochran54, mantelhaenszel59}, which  can be adapted for continuous random variables by binning the data~\cite{huang10} or using kernels~\cite{fukumizuetal07, zhangetal12, strobletal19, doranetal14, senetal17}.  Besides these, there are methods based on estimating conditional cumulative distribution functions~\cite{lintongozalo97, patraetal16},  conditional characteristic functions~\cite{suwhite07, keyin19}, conditional probability density functions~\cite{suwhite08}, empirical likelihood~\cite{suwhite14}, mutual information and entropy~\cite{runge18, joe89, poczosschneider12}, copulas~\cite{bergsma04, song09, veraverbekeetal11}, distance correlation~\cite{wangetal15, fanetal20, szekelyrizzo14}, and other approaches~\cite{sethprincipe12}. A number of interesting ideas based on resampling and permutation tests have been proposed in recent years~\cite{candesetal18, senetal17, berrettetal20}.

In  \citet{azadkiachatterjee21}, a new coefficient of conditional dependence was proposed, based on the ideas behind the $\xi$-coefficient defined in \cite{chatterjee21a}. Like the $\xi$-coefficient, this one also has a long list of desirable features, such as being fully nonparametric and working under minimal assumptions. The coefficient is defined as follows. 

Let $Y$ be a random variable and $\bbx = (X_1,\ldots,X_p)$ and $\bbz = (Z_1,\ldots, Z_q)$ be random vectors, all defined on the same probability space. Here $q\ge 1$ and $p\ge 0$. The value $p=0$ means that $\bbx$ has no components at all.  Let $\mu$ be the law of $Y$. The following quantity was proposed in \cite{azadkiachatterjee21} as a measure of the degree of conditional dependence of $Y$ and $\bbz$ given $\bbx$:
\begin{equation}\label{tdef0}
T = T(Y,\bbz|\bbx) := \frac{\int \ee(\var(\pp(Y\ge t|\bbz, \bbx)|\bbx)) d\mu(t)}{\int \ee(\var(1_{\{Y\ge t\}}|\bbx))d\mu(t)}.
\end{equation}
If the denominator equals zero, $T$ is undefined.  If $p=0$, then $\bbx$ has no components, and the conditional expectations and variances given $\bbx$ should be interpreted as unconditional expectations and variances. In this case we write $T(Y,\bbz)$ instead of $T(Y,\bbz|\bbx)$. Note that $T$ is a generalization of the statistic $\xi$ appearing in Theorem \ref{mainthm}. The following theorem summarizes the main properties of~$T$.
\begin{thm}[\cite{azadkiachatterjee21}]\label{popthm}
Suppose that $Y$ is not almost surely equal to a measurable function of $\bbx$ (when $p=0$, this means that $Y$ is not almost surely a constant). Then $T$ is well-defined and $0\le T \le 1$. Moreover, $T= 0$ if and only if $Y$ and $\bbz$ are conditionally independent given $\bbx$, and $T=1$ if and only if $Y$  is almost surely equal to a measurable function of $\bbz$ given $\bbx$. When $p=0$, conditional independence given $\bbx$ simply means unconditional independence. 
\end{thm}
Now suppose we have data consisting of $n$ i.i.d.~copies $(Y_1,\bbx_1,\bbz_1),\ldots,(Y_n, \bbx_n,\bbz_n)$ of the triple $(Y,\bbx,\bbz)$, where $n\ge 2$. For each $i$, let $N(i)$ be the index $j$ such that $\bbx_j$ is the nearest neighbor of $\bbx_i$ with respect to the Euclidean metric on $\rr^p$, where ties are broken uniformly at random. Let $M(i)$ be the index $j$ such that $(\bbx_j, \bbz_j)$ is the nearest neighbor of $(\bbx_i, \bbz_i)$ in $\rr^{p+q}$, again with ties broken uniformly at random. Let $R_i$ be the rank of $Y_i$, that is, the number of $j$ such that $Y_j\le Y_i$. If $p\ge 1$, define 
\[
T_n = T_n(Y, \bbz|\bbx) := \frac{\sum_{i=1}^n (\min\{R_i, R_{M(i)}\} - \min\{R_i, R_{N(i)}\})}{\sum_{i=1}^n (R_i - \min\{R_i, R_{N(i)}\})}.
\]
If $p=0$, let $L_i$ be the number of $j$ such that $Y_j\ge Y_i$, let $M(i)$ denote the $j$ such that $\bbz_j$ is the nearest neighbor of $\bbz_i$ (ties broken uniformly at random), and let
\[
T_n = T_n(Y, \bbz) := \frac{\sum_{i=1}^n (n\min\{R_i, R_{M(i)}\} - L_i^2)}{\sum_{i=1}^n L_i(n-L_i)}.
\]
In both cases, $T_n$ is undefined if the denominator is zero. 
The following theorem shows that $T_n$ is a consistent estimator of $T$.
\begin{thm}[\cite{azadkiachatterjee21}]\label{sampthm}
Suppose that $Y$ is not almost surely equal to a measurable function of $\bbx$. Then as  $n\to\infty$, $T_n \to T$ almost surely.
\end{thm}
For various other properties of $T_n$, such as rate of convergence, performance in simulations and real data, etc., see \cite{azadkiachatterjee21}. One problem that was left unsolved in \cite{azadkiachatterjee21} was the question of proving a central limit theorem for $T_n$ under the null hypothesis, which is crucial for carrying out tests for conditional independence. This question was partially resolved by \citet*{shietal21}, who proved a central limit theorem for $T_n$ under the assumption that $Y$ is independent of $(\bbx,\bbz)$. An improved version of this result was proved recently by \citet{linhan22}. A version for data supported on manifolds was proved by  \citet{hanhuang22}. 

The paper of \citet{shietal21} also develops the `conditional randomization test' (CRT) framework of \citet{candesetal18} to test conditional independence using $T_n$, and find that $T_n$, like $\xi_n$ is an inefficient test statistic. To address this concern, an improved generalization of $T_n$, called `kernel partial correlation' (KPC), was proposed by \citet*{huangetal20}. Unlike $T_n$, KPC has the flexibility to use more than one nearest neighbor, which gives it better power properties.

Note that by the above theorems, a test of conditional independence based on $T_n$ is consistent against all alternatives. The problem is that in the absence of an asymptotic theory for $T_n$, it is difficult to control the significance level of such a test. This is in fact an impossible problem, by a recent result of \citet{shahpeters20} that proves hardness of conditional independence testing in the absence of smoothness assumptions. Assuming some degree of smoothness, minimax optimal conditional independence tests were recently constructed by \citet*{neykovetal21} and \citet{kimetal21}.

\section{Application to nonparametric variable selection}
The commonly used variable selection methods in the statistics literature rely on linear or  additive models. This includes classical methods~\cite{breiman95, georgemcculloch93, chendonoho94, tibshirani96, efronetal04, friedman91, hastieetal09, miller02} as well as modern ones~\cite{candestao07,   zou06, zouhastie05, yuanlin06, fanli01,  ravikumaretal09}. These methods are powerful and widely used in practice. However, they sometimes run  into problems when significant interaction effects or nonlinearities are present. Such problems can be overcome by model-free methods~\cite{candesetal18, ho98, amitgeman97, breiman96, freundschapire96, hastieetal09, breimanetal84, battiti94, vergaraestevez14, breiman01}. On the flip side, the  theoretical foundations of model-free methods are usually weaker than those of model-based methods.

In an attempt to combine the best of both worlds, a new method of variable selection called Feature Ordering by Conditional Independence~(FOCI), was proposed in \citet{azadkiachatterjee21}. This method uses the  conditional dependence coefficient $T_n$ described in the previous section in a stepwise fashion, as follows. Let $Y$ be the response variable and let $\bbx = (X_j)_{1\le j\le p}$ be the set of predictors. The data consists of $n$ i.i.d.~copies of $(Y,\bbx)$. First, choose $j_1$ to be the index $j$ that maximizes $T_n(Y, X_j)$. Having obtained $j_1,\ldots, j_k$,  choose $j_{k+1}$ to be the index $j\notin\{j_1,\ldots, j_{k}\}$ that maximizes $T_n(Y, X_j|X_{j_1},\ldots,X_{j_{k}})$. Continue like this until arriving at the first $k$ such that $T_n(Y, X_{j_{k+1}}|X_{j_1},\ldots,X_{j_{k}})\le 0$, and then declare the chosen subset to be $\hat{S} := \{j_1,\ldots,j_k\}$. If there is no such $k$, define $\hat{S}$ to be the whole set of variables. It may also happen that $T_n(Y, X_{j_1})\le 0$. In that case declare $\hat{S}$ to be empty. Note that this variable selection procedure involves no choice of tuning parameters, which may be an advantage in practice.

It was shown in \cite{azadkiachatterjee21} that under mild conditions, the method selects a `correct' set of variables with high probability. More precisely, it was shown that with high probability, the set $\hat{S}$ selected by FOCI has the property that $Y$ and $(X_j)_{j\notin S}$ are conditionally independent given $(X_j)_{j\in \hat{S}}$. In other words, all the information about $Y$ that one can get from $\bbx$ is contained in $(X_j)_{j\in \hat{S}}$. For further properties of FOCI and its performance in simulations and real data sets, see \cite{azadkiachatterjee21}.

An improved version of FOCI called KFOCI (`Kernel FOCI') was proposed by \citet{huangetal20}. An application of FOCI to causal inference, via an algorithm named DAG-FOCI, was introduced in \citet*{azadkiaetal21}. For another application to causal inference, see \citet{chatterjeevidyasagar22}.


\section{A new proposal: Generalization to standard Borel spaces}\label{newsec}
In this section, a simple but wide ranging generalization of $\xi_n$ is proposed. In hindsight, this generalization seems obvious, but somehow  this was overlooked both in the original paper \cite{chatterjee21a} as well as in the subsequent developments listed in Section \ref{multsec}. 

Recall that two measurable spaces are said to be isomorphic to each other if there is a bijection between the two spaces which is measurable and whose inverse is also measurable. Recall that a {\it standard Borel space} is a measurable space that is isomorphic to a Borel subset of a Polish space~\cite[Chapter 3]{srivastava98}. In particular, every Borel subset of every Polish space is a standard Borel space. The Borel isomorphism theorem says that any uncountable standard Borel space is isomorphic to the real line (see \citet{raosrivastava94} for an elementary proof). In particular, if $\mx$ is any standard Borel space, there is a measurable map $\varphi:\mx \to \rr$ such that $\varphi$ is injective, $\varphi(\mx)$ is Borel, and $\varphi^{-1}$ is measurable on $\varphi(\mx)$. We will say that $\varphi$ is an isomorphism between $\mx$ and a Borel subset of the real line, or simply, a `Borel isomorphism'.

Now let $\mx$ and $\my$ be two standard Borel spaces. Let $\varphi$ be a Borel isomorphism  of  $\mx$ and $\psi$ be a Borel isomorphism of $\my$. Let $(X,Y)$ be an $\mx\times \my$-valued pair of random variables, and let $(X_1,Y_1),\ldots,(X_n,Y_n)$ be i.i.d.~copies of $(X,Y)$. Let $X':= \varphi(X)$ and $Y' := \psi(Y)$, so that $(X',Y')$ is a pair of real-valued random variables. Let $X_i' := \varphi(X_i)$ and $Y_i' := \psi(Y_i)$ for each $i$. Finally, define
\[
\xi_n(X,Y) := \xi_n(X',Y'),
\]
where $\xi_n(X',Y')$ is defined using $(X_1',Y_1'),\ldots,(X_n',Y_n')$ as in equation \eqref{xindef}. Note that the definition of $\xi_n(X,Y)$ depends on our choices of $\varphi$ and $\psi$. Different choices of isomorphisms would lead different definitions of $\xi_n$. The following theorem generalizes Theorem \ref{mainthm}. 
\begin{thm}\label{mainthm2}
If $Y$ is not almost surely a constant, then as $n\to\infty$, $\xi_n(X,Y)$ converges almost surely to the deterministic limit $\xi(X,Y)$, which equals $\xi(X',Y')$, defined as in \eqref{xidef} with $X'$ and $Y'$ in place of $X$ and $Y$. This limit belongs to the interval $[0,1]$. It is  $0$ if and only if $X$ and $Y$ are independent, and it is $1$ if and only if there is a measurable function $f:\rr\to\rr$ such that $Y=f(X)$ almost surely. Moreover, the asymptotic distribution of $\sqrt{n}\xi_n(X,Y)$ under the hypothesis of independence, as given by Theorems \ref{cltthm} and \ref{estthm}, also holds, provided that $\tau$ and $\hat{\tau}_n$ are computed using $X_i'$ and $Y_i'$  instead of $X_i$ and $Y_i$.
\end{thm}
\begin{proof}
The convergence is clear from Theorem \ref{mainthm}. Also, by Theorem \ref{mainthm}, $\xi(X',Y') = 0$  if and only if $X'$ and $Y'$ are independent, and $\xi(X',Y')=1$ if and only if $Y'$ is a measurable function of $X'$. Since $X= \varphi^{-1}(X')$ and $Y = \psi^{-1}(Y')$, it follows that $X$ and $Y$ are independent if and only if $X'$ and $Y'$ are independent. For the same reason, $Y$ is a measurable function of $X$ almost surely if and only if $Y'$ is a measurable function of $X'$ almost surely. Note that $Y$ is not almost surely a constant if and only if $Y'$ is not almost surely a constant. Lastly, since $\xi_n(X,Y)$ is just $\xi_n(X',Y')$, any result about the asymptotic distribution of $\xi_n(X',Y')$, including Theorems \ref{cltthm} and \ref{estthm} of this draft, can be transferred to $\xi_n(X,Y)$.
\end{proof}

Just like the univariate case, the generalized $\xi_n$ has the advantage of working under zero assumptions and having a simple asymptotic theory, as shown by Theorem \ref{mainthm2}. On the other hand, just like the univariate coefficient, one can expect the generalized coefficient to also suffer from low power for testing independence. 

Theorem \ref{mainthm2} is a nice, clean result, but to implement the idea in practice, one needs to  work with  actual Borel isomorphisms. Here is an example of a Borel isomorphism between $\rr^d$ and a Borel subset of $\rr$. Take any $x= (x_1,\ldots,x_d)\in \rr^d$. Let 
\[
a_{i,1}\cdots a_{i,k_i} . b_{i,1}b_{i,2}\cdots
\]
be the binary expansion of $|x_i|$. Filling in extra $0$'s at the beginning if necessary, let us assume that $k_1=\cdots=k_d = k$. Then, let us `interlace' the digits to get the number
\[
a_{1,1}a_{2,1}\cdots a_{d,1} a_{1,2}a_{2,2}\cdots a_{d,2} \cdots a_{1,k}a_{2,k}\cdots a_{d,k}.b_{1,1}b_{2,1}\cdots b_{d,1}b_{1,2}b_{2,2}\cdots b_{d,2}\cdots.
\]
This is an encoding of the $n$-tuple $(|x_1|,\ldots,|x_d|)$. But we also want to encode the signs of the $x_i$'s. Let $c_i = 1$ if $x_i\ge 0$ and $0$ if $x_i < 0$. Sticking $1c_1c_2\cdots c_d$ in front of the above list, we get an encoding of the vector $x$ as a real number. (The $1$ in front ensures there is no ambiguity arising from some of the leading $c_i$'s being $0$.) It is easy to verify that this mapping is measurable, injective, and its inverse (defined on its range) is also measurable.

Numerical simulations with $\xi_n$ computed using the above scheme produced satisfactory results. Some examples are as follows. The examples show one potential problem with using statistics such as $\xi_n$ in a high dimensional setting: The bias may be quite large (even though the standard deviation is small), resulting in slow convergence of $\xi_n(X,Y)$ to $\xi(X,Y)$. 
\begin{ex}[Points on a sphere]\label{exsphere}
Non-uniform random points were generated on the unit sphere in $\rr^3$ by drawing $\phi_1,\ldots,\phi_n$ uniformly from $[-\pi, \pi]$, drawing $\theta_1,\ldots, \theta_n$ uniformly from $[0,2\pi]$, and defining
\[
x_i = \sin \phi_i \cos\theta_i, \ \  y_i = \sin \phi_i \sin \theta_i, \ \ z_i = \cos\phi_i.
\]
Taking $X_i = (\phi_i, \theta_i)$ and $Y_i = (x_i, y_i, z_i)$, $\xi_n(X,Y)$ was computed for $n=100$ and $n=1000$. One thousand simulations were done for each $n$. The histograms of the values of $\xi_n$ are displayed in Figure \ref{sphere1}. For $n=100$, $\xi_n(X,Y)$ had a mean value of $0.617$ with a standard deviation of $0.049$. For $n=1000$, the mean value was $0.865$ and the standard deviation was $0.009$. Simulations were also done for $n=10000$, where the mean value of $\xi_n(X,Y)$ turned out to be $0.957$ and the standard deviation was $0.002$. The slow convergence due to large bias is clearly apparent in this example.
\begin{figure}[t]
\begin{center}
\begin{subfigure}{.48\textwidth}
\includegraphics[width = \textwidth]{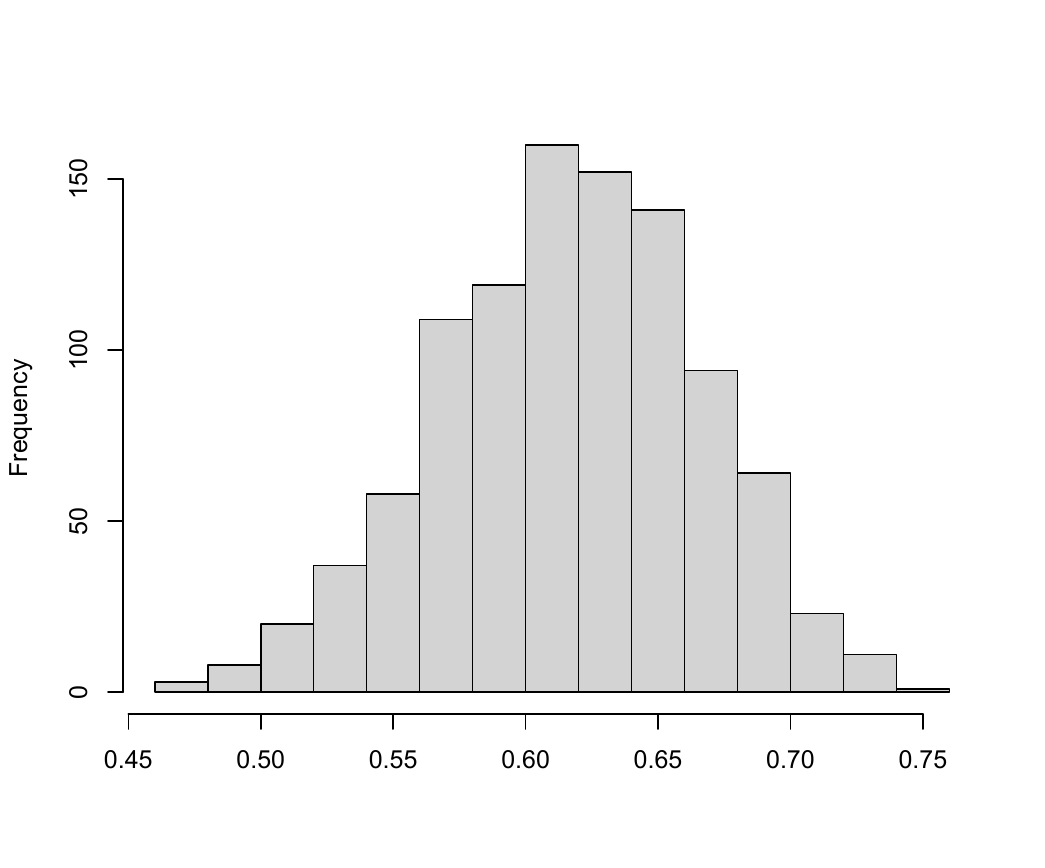}
\caption{$n=100$}
\end{subfigure}
\begin{subfigure}{.48\textwidth}
\includegraphics[width = \textwidth]{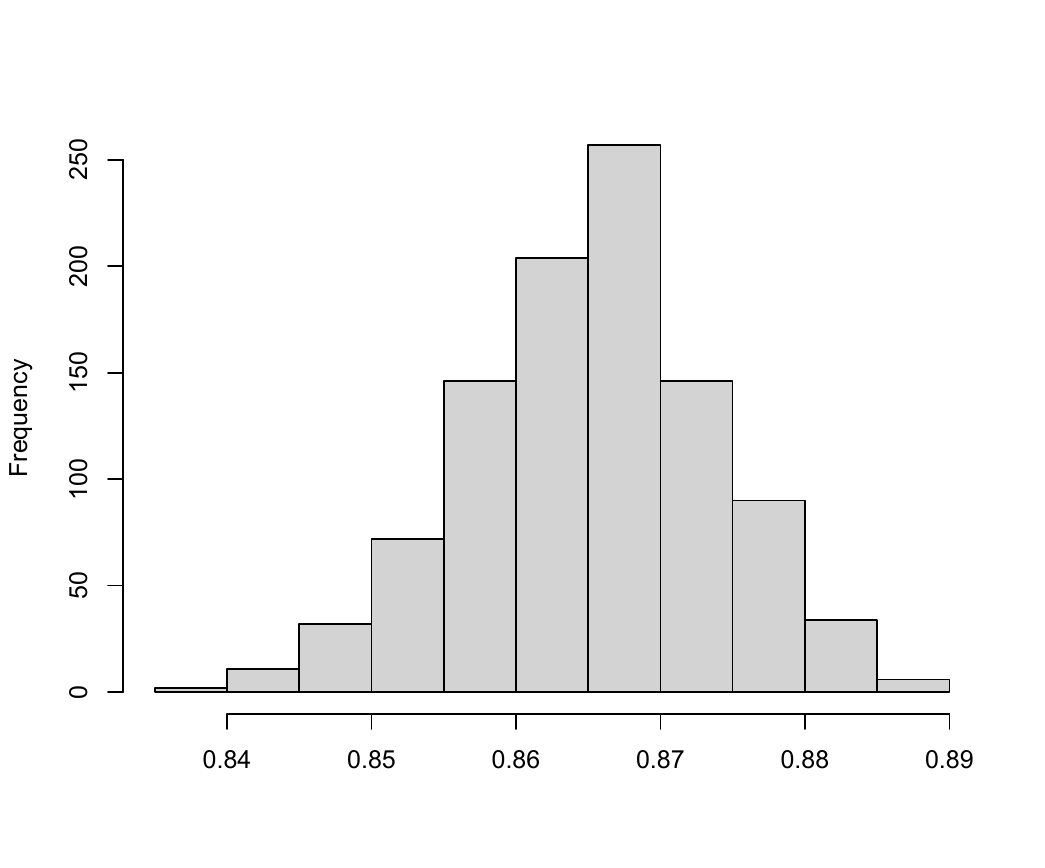}
\caption{$n=1000$}
\end{subfigure}
\caption{Histograms of values of $\xi_n(X,Y)$ when $X$ and $Y$ are the polar and Cartesian coordinates of random points from the unit sphere in $\rr^3$ (Example \ref{exsphere}).\label{sphere1}}
\end{center}
\end{figure}
\end{ex}

\begin{ex}[Points on a sphere plus noise]\label{exsphere2}
This is the same example as the previous one, except that independent random noises were added to $x_i$, $y_i$ and $z_i$. The noise variables were taken to be normal with mean zero and standard deviation $\sigma$, for various values of $\sigma$. Table \ref{sphere2} displays the means and standard deviations of $\xi_n(X,Y)$ in $1000$ simulations, for $n=100$ and $n=1000$, and $\sigma=0.01$, $\sigma=0.05$ and $\sigma = 0.1$.  
\begin{table}[t]
\begin{center}
\caption{Means and standard deviations of $\xi_n(X,Y)$ in $1000$ simulations, when $X=$ polar coordinates of a random point on a sphere and $Y = $ Cartesian coordinates of the point plus independent $N(0,\sigma^2)$ errors in each coordinate (Example \ref{exsphere2}). \label{sphere2}}
\begin{tabular}{rrrr}
\toprule
&\multicolumn{3}{c}{$\sigma$}\\
\cmidrule{2-4}
$n$ & $0.01$ & $0.05$ & $0.1$ \\
\midrule
100 & $0.545$ ($0.062$) & $0.440$ ($0.071$) & $0.357$ ($0.074$)\\
1000 & $0.783$ ($0.015$) & $0.658$ ($0.020$) & $0.543$ ($0.023$)\\
\bottomrule
\end{tabular}
\end{center}
\end{table}
\end{ex}

\begin{ex}[Marginal independence versus joint dependence]\label{exjoint}
In this example, we have a pair of random variables $Y= (a, b)$ which is a function of a $4$-tuple $X = (u,v,w,z)$, but $u$, $v$, $w$ and $z$ are marginally independent of $Y_i$. The variables are constructed as follows. Let $a$, $b$ and $c$ be independent Uniform$[0, 1]$ random variables. Let 
\begin{align*}
&u = a + b + c \ \ (\mathrm{mod} \ 1), \ \ \ v = \frac{a}{2} + \frac{b}{2}  + c \ \ (\mathrm{mod} \ 1),\\
&w = \frac{4a}{3} + \frac{2b}{3}  + c \ \ (\mathrm{mod} \ 1), \ \ \ z = \frac{2a}{3} + \frac{b}{3}  + c \ \ (\mathrm{mod} \ 1).
\end{align*}
Here $x \ (\mathrm{mod} \ 1)$ denotes the `remainder modulo $1$' of a real number $x$. It is easy to see that individually, $u$, $v$, $w$ and $z$ are independent of $Y = (a,b)$, since the operation of adding the uniform random variable $c$ and taking the remainder modulo $1$ erases all information about the quantity to which $c$ is added. However, $Y$ can be recovered from the $4$-tuple $X = (u,v,w,z)$ because $u-v = (a+b)/2 \ (\mathrm{mod} \ 1) = (a+b)/2$ and $w-z = (2a + b)/3 \ (\mathrm{mod} \ 1) = (2a+b)/3$, where the second identity holds in each case because $(a+b)/2\in [0,1]$ and $(2a+b)/3\in [0,1]$. With the above definitions, $n$ i.i.d.~copies of $(X,Y)$ were sampled, and $\xi_n(X,Y)$ was computed, along with the asymptotic P-value for testing independence. This was repeated one thousand times. The average values of the coefficients and the P-values for $n=100$ and $n=1000$ are reported in Table \ref{joint}. The table shows that even for $n$ as small as $100$, the hypothesis that $X$ and $Y$ are independent is rejected with average P-value $0.001$, whereas the average P-value for the hypothesis and $u$ and $Y$ are independent is $0.491$. On the other hand, even for $n$ as large as $1000$, the average P-value for the hypothesis that $u$ and $Y$ are independent is $0.495$, whereas the average P-value for the hypothesis that $X$ and $Y$ are independent is $0.000$. 
\begin{table}[t]
\begin{center}
\caption{Average values of $\xi_n(u,Y)$ and $\xi_n(X,Y)$ and average P-values for testing independence in Example \ref{exjoint}.\label{joint}}
\begin{tabular}{rrrrrr}
\toprule
&\multicolumn{2}{c}{Average value}& & \multicolumn{2}{c}{Average P-value}\\
\cmidrule{2-3}\cmidrule{5-6}
$n$ & $\xi_n(u,Y)$ & $\xi_n(X,Y)$ & & $H_0: u \perp\!\!\!\perp Y$ & $H_0: X \perp\!\!\!\perp Y$  \\
\midrule
$100$ & $0.002$ & $0.313$ & & $0.491$ & $0.001$\\
$1000$ & $0.001$ & $0.581$ & & $0.495$ & $0.000$\\
\bottomrule
\end{tabular}
\end{center}
\end{table}
\end{ex}

The above examples indicate that it may not be `crazy' to use the generalized version of $\xi_n$ as a measure of association in the multivariate setting and beyond. Further investigations, through applications in simulated and real data, would be necessary to arrive at a concrete verdict. 

The generalized version of $\xi_n$ can also be used to define a coefficient of conditional dependence for random variables taking values in standard Borel spaces, as follows. Let $\mx$, $\my$ and $\mz$ be standard Borel spaces, and $X$, $Y$ and $Z$ be random variables taking values in $\mx$, $\my$ and $\mz$, respectively. Let $W$ denote the $\mx \times \mz$-valued random variable $(X,Z)$. Using Borel isomorphisms, let us define the coefficients $\xi_n(X,Y)$ and $\xi_n(W,Y)$ as before, based on i.i.d.~data $(X_1,Y_1,Z_1),\ldots,(X_n, Y_n,Z_n)$. One can then define a coefficient of conditional dependence between $Y$ and $Z$ given $X$ as 
\[
\xi_n(Z,Y|X) := \frac{\xi_n(W,Y) - \xi_n(X,Y)}{1- \xi_n(X,Y)},
\]
leaving it undefined if the denominator is zero. The following theorem justifies its use as a measure of conditional dependence.
\begin{thm}
Suppose that $Y$ is not almost surely equal to a measurable function of $X$. Then as $n\to \infty$, $\xi_n(Z,Y|X)$ converges almost surely to a deterministic limit $\xi(Z,Y|X)\in [0,1]$, which is $0$ if and only if $Y$ and $Z$ are independent given $X$, and $1$ if and only if $Y$ is almost surely equal to a measurable function of $Z$ given $X$.
\end{thm}
\begin{proof}
Let $\varphi$, $\psi$, $\eta$ and $\delta$ be the Borel isomorphisms of $\mx$, $\my$, $\mz$ and $\mx \times \mz$ that we are using in our definitions of $\xi_n$. Let $X' := \varphi(X)$, $Y' := \psi(Y)$, $Z' := \eta(Z)$ and $W' := \delta(W)$. By Theorem \ref{mainthm} and the assumption  that $Y$ is not almost surely a measurable function of $X$, we get that $\xi_n(Z,Y|X)$ converges almost surely to
\begin{align*}
\xi(Z,Y|X) &:= \frac{\xi(W',Y') - \xi(X',Y')}{1-\xi(X',Y')}. 
\end{align*}
Let $\mu$ denote  the law of $Y'$. Then 
\begin{align*}
\xi(W',Y') - \xi(X',Y') &= \frac{\int \var(\pp(Y'\ge t |W')) d\mu(t) - \int \var(\pp(Y'\ge t |X')) d\mu(t)}{\int \var(1_{\{Y'\ge t\}}) d\mu(t)} \\
&=  \frac{\int \var(\pp(Y'\ge t|X', Z')) d\mu(t) - \int \var(\pp(Y'\ge t|X')) d\mu(t)}{\int \var(1_{\{Y'\ge t\}}) d\mu(t)} \\
&= \frac{\int \ee(\var(\pp(Y'\ge t |X', Z')|X') d\mu(t)}{\int \var(1_{\{Y'\ge t\}}) d\mu(t)}.
\end{align*}
Similarly,
\[
1-\xi(X',Y') = \frac{\int \ee(\var(1_{\{Y'\ge t \}}|X')) d\mu(t)}{\int \var(1_{\{Y'\ge t\}}) d\mu(t)}.
\]
From the above expressions, we see that $\xi(Z,Y|X)$ is nothing but the quantity $T(Y,Z|X)$ displayed in equation \eqref{tdef0}. The claims of the theorem now follow from Theorem \ref{popthm}.
\end{proof}

\section{R packages}
An R package for calculating $\xi_n$, as well as the generalized coefficient proposed in this survey, and P-values for testing independence --- named XICOR --- is available on CRAN~\cite{chatterjeeholmes20}.  An R package for calculating the conditional dependence coefficient $T_n$ and implementing the FOCI algorithm, called FOCI, is also available~\cite{azadkiaetal20}. The KFOCI algorithm of \citet{huangetal20} is implemented in an R package by the same name~\cite{huangetal22}.




\bibliographystyle{abbrvnat}

\bibliography{myrefs}

\begin{thebibliography}{133}
\providecommand{\natexlab}[1]{#1}
\providecommand{\url}[1]{\texttt{#1}}
\expandafter\ifx\csname urlstyle\endcsname\relax
  \providecommand{\doi}[1]{doi: #1}\else
  \providecommand{\doi}{doi: \begingroup \urlstyle{rm}\Url}\fi

\bibitem[Amit and Geman(1997)]{amitgeman97}
Y.~Amit and D.~Geman.
\newblock Shape quantization and recognition with randomized trees.
\newblock \emph{Neural Computation}, 9\penalty0 (7):\penalty0 1545--1588, 1997.

\bibitem[Auddy et~al.(2021)Auddy, Deb, and Nandy]{auddyetal21}
A.~Auddy, N.~Deb, and S.~Nandy.
\newblock {Exact Detection Thresholds for Chatterjee's Correlation}.
\newblock \emph{arXiv preprint arXiv:2104.15140}, 2021.

\bibitem[Azadkia and Chatterjee(2021)]{azadkiachatterjee21}
M.~Azadkia and S.~Chatterjee.
\newblock A simple measure of conditional dependence.
\newblock \emph{Annals of Statistics}, 49\penalty0 (6):\penalty0 3070--3102,
  2021.

\bibitem[Azadkia et~al.(2020)Azadkia, Chatterjee, and Matloff]{azadkiaetal20}
M.~Azadkia, S.~Chatterjee, and N.~S. Matloff.
\newblock \emph{FOCI: Feature Ordering by Conditional Independence}, 2020.
\newblock URL \url{https: //CRAN.R-project.org/package=FOCI}.

\bibitem[Azadkia et~al.(2021)Azadkia, Taeb, and B{\"u}hlmann]{azadkiaetal21}
M.~Azadkia, A.~Taeb, and P.~B{\"u}hlmann.
\newblock A fast non-parametric approach for causal structure learning in
  polytrees.
\newblock \emph{arXiv preprint arXiv:2111.14969}, 2021.

\bibitem[Battiti(1994)]{battiti94}
R.~Battiti.
\newblock Using mutual information for selecting features in supervised neural
  net learning.
\newblock \emph{IEEE Transactions on Neural Networks}, 5\penalty0 (4):\penalty0
  537--550, 1994.

\bibitem[Bergsma(2004)]{bergsma04}
W.~Bergsma.
\newblock Testing conditional independence for continuous random variables.
\newblock \emph{Report Eurandom}, 2004048, 2004.

\bibitem[Bergsma and Dassios(2014)]{bergsmadassios14}
W.~Bergsma and A.~Dassios.
\newblock A consistent test of independence based on a sign covariance related
  to {K}endall’s tau.
\newblock \emph{Bernoulli}, 20\penalty0 (2):\penalty0 1006--1028, 2014.

\bibitem[Berrett et~al.(2020)Berrett, Wang, Barber, and
  Samworth]{berrettetal20}
T.~B. Berrett, Y.~Wang, R.~F. Barber, and R.~J. Samworth.
\newblock The conditional permutation test for independence while controlling
  for confounders.
\newblock \emph{Journal of the Royal Statistical Society: Series B (Statistical
  Methodology)}, 82\penalty0 (1):\penalty0 175--197, 2020.

\bibitem[Berrett et~al.(2021)Berrett, Kontoyiannis, and
  Samworth]{berrettetal21}
T.~B. Berrett, I.~Kontoyiannis, and R.~J. Samworth.
\newblock Optimal rates for independence testing via $u$-statistic permutation
  tests.
\newblock \emph{Annals of Statistics}, 49\penalty0 (5):\penalty0 2457--2490,
  2021.

\bibitem[Bhattacharya(2019)]{bhattacharya19}
B.~B. Bhattacharya.
\newblock A general asymptotic framework for distribution-free graph-based
  two-sample tests.
\newblock \emph{Journal of the Royal Statistical Society: Series B (Statistical
  Methodology)}, 81\penalty0 (3):\penalty0 575--602, 2019.

\bibitem[Bickel(2022)]{bickel22}
P.~J. Bickel.
\newblock Measures of independence and functional dependence.
\newblock \emph{arXiv preprint arXiv:2206.13663}, 2022.

\bibitem[Blum et~al.(1961)Blum, Kiefer, and Rosenblatt]{blumetal61}
J.~Blum, J.~Kiefer, and M.~Rosenblatt.
\newblock Distribution free tests of independence based on the sample
  distribution function.
\newblock \emph{Annals of Mathematical Statistics}, 32\penalty0 (2):\penalty0
  485--498, 1961.

\bibitem[Breiman(1995)]{breiman95}
L.~Breiman.
\newblock Better subset regression using the nonnegative garrote.
\newblock \emph{Technometrics}, 37\penalty0 (4):\penalty0 373--384, 1995.

\bibitem[Breiman(1996)]{breiman96}
L.~Breiman.
\newblock Bagging predictors.
\newblock \emph{Machine Learning}, 24\penalty0 (2):\penalty0 123--140, 1996.

\bibitem[Breiman(2001)]{breiman01}
L.~Breiman.
\newblock Random forests.
\newblock \emph{Machine Learning}, 45\penalty0 (1):\penalty0 5--32, 2001.

\bibitem[Breiman and Friedman(1985)]{breimanfriedman85}
L.~Breiman and J.~H. Friedman.
\newblock Estimating optimal transformations for multiple regression and
  correlation.
\newblock \emph{Journal of the American statistical Association}, 80\penalty0
  (391):\penalty0 580--598, 1985.

\bibitem[Breiman et~al.(1984)Breiman, Friedman, Olshen, and
  Stone]{breimanetal84}
L.~Breiman, J.~H. Friedman, R.~A. Olshen, and C.~J. Stone.
\newblock \emph{Classification and Regression Trees}.
\newblock Wadsworth Press, 1984.

\bibitem[Cand\`es and Tao(2007)]{candestao07}
E.~Cand\`es and T.~Tao.
\newblock {The Dantzig Selector: Statistical estimation when $p$ is much larger
  than $n$}.
\newblock \emph{Annals of Statistics}, 35\penalty0 (6):\penalty0 2313--2351,
  2007.

\bibitem[Cand\`es et~al.(2018)Cand\`es, Fan, Janson, and Lv]{candesetal18}
E.~Cand\`es, Y.~Fan, L.~Janson, and J.~Lv.
\newblock Panning for gold: `model-{X}' knockoffs for high dimensional
  controlled variable selection.
\newblock \emph{Journal of the Royal Statistical Society: Series B (Statistical
  Methodology)}, 80\penalty0 (3):\penalty0 551--577, 2018.

\bibitem[Cao and Bickel(2020)]{caobickel20}
S.~Cao and P.~J. Bickel.
\newblock Correlations with tailored extremal properties.
\newblock \emph{arXiv preprint arXiv:2008.10177}, 2020.

\bibitem[Chatterjee(2021)]{chatterjee21a}
S.~Chatterjee.
\newblock A new coefficient of correlation.
\newblock \emph{Journal of the American Statistical Association}, 116\penalty0
  (536):\penalty0 2009--2022, 2021.

\bibitem[Chatterjee and Holmes(2020)]{chatterjeeholmes20}
S.~Chatterjee and S.~Holmes.
\newblock \emph{XICOR: Association measurement through cross rank increments},
  2020.
\newblock URL \url{https://CRAN.R-project.org/package=XICOR}.

\bibitem[Chatterjee and Vidyasagar(2022)]{chatterjeevidyasagar22}
S.~Chatterjee and M.~Vidyasagar.
\newblock Estimating large causal polytree skeletons from small samples.
\newblock \emph{arXiv preprint arXiv:2209.07028}, 2022.

\bibitem[Chen(2020)]{chen20}
L.-P. Chen.
\newblock A note of feature screening via rank-based coefficient of
  correlation.
\newblock \emph{arXiv preprint arXiv:2008.04456}, 2020.

\bibitem[Chen and Donoho(1994)]{chendonoho94}
S.~Chen and D.~Donoho.
\newblock Basis pursuit.
\newblock In \emph{Proceedings of 1994 28th Asilomar Conference on Signals,
  Systems and Computers}, volume~1, pages 41--44. IEEE, 1994.

\bibitem[Chernozhukov et~al.(2017)Chernozhukov, Galichon, Hallin, and
  Henry]{chernozhukovetal17}
V.~Chernozhukov, A.~Galichon, M.~Hallin, and M.~Henry.
\newblock {Monge--Kantorovich depth, quantiles, ranks and signs}.
\newblock \emph{Annals of Statistics}, 45\penalty0 (1):\penalty0 223--256,
  2017.

\bibitem[Cochran(1954)]{cochran54}
W.~G. Cochran.
\newblock Some methods for strengthening the common $\chi^2$ tests.
\newblock \emph{Biometrics}, 10\penalty0 (4):\penalty0 417--451, 1954.

\bibitem[Cs{\"o}rg{\H{o}}(1985)]{csorgo85}
S.~Cs{\"o}rg{\H{o}}.
\newblock Testing for independence by the empirical characteristic function.
\newblock \emph{Journal of Multivariate Analysis}, 16\penalty0 (3):\penalty0
  290--299, 1985.

\bibitem[Deb and Sen(2021)]{debsen21}
N.~Deb and B.~Sen.
\newblock Multivariate rank-based distribution-free nonparametric testing using
  measure transportation.
\newblock \emph{Journal of the American Statistical Association}, pages 1--16,
  2021.

\bibitem[Deb et~al.(2020)Deb, Ghosal, and Sen]{debetal20}
N.~Deb, P.~Ghosal, and B.~Sen.
\newblock Measuring association on topological spaces using kernels and
  geometric graphs.
\newblock \emph{arXiv preprint arXiv:2010.01768}, 2020.

\bibitem[Dette et~al.(2013)Dette, Siburg, and Stoimenov]{detteetal13}
H.~Dette, K.~F. Siburg, and P.~A. Stoimenov.
\newblock A copula-based non-parametric measure of regression dependence.
\newblock \emph{Scandinavian Journal of Statistics}, 40\penalty0 (1):\penalty0
  21--41, 2013.

\bibitem[Doran et~al.(2014)Doran, Muandet, Zhang, and
  Sch{\"o}lkopf]{doranetal14}
G.~Doran, K.~Muandet, K.~Zhang, and B.~Sch{\"o}lkopf.
\newblock A permutation-based kernel conditional independence test.
\newblock In \emph{Uncertainty in Artificial Intelligence}, pages 132--141.
  AUAI, 2014.

\bibitem[Drton et~al.(2020)Drton, Han, and Shi]{drtonetal20}
M.~Drton, F.~Han, and H.~Shi.
\newblock High-dimensional consistent independence testing with maxima of rank
  correlations.
\newblock \emph{Annals of Statistics}, 48\penalty0 (6):\penalty0 3206--3227,
  2020.

\bibitem[Efron et~al.(2004)Efron, Hastie, Johnstone, and
  Tibshirani]{efronetal04}
B.~Efron, T.~Hastie, I.~Johnstone, and R.~Tibshirani.
\newblock Least angle regression.
\newblock \emph{Annals of statistics}, 32\penalty0 (2):\penalty0 407--499,
  2004.

\bibitem[Fan and Li(2001)]{fanli01}
J.~Fan and R.~Li.
\newblock Variable selection via nonconcave penalized likelihood and its oracle
  properties.
\newblock \emph{Journal of the American Statistical Association}, 96\penalty0
  (456):\penalty0 1348--1360, 2001.

\bibitem[Fan et~al.(2020)Fan, Feng, and Xia]{fanetal20}
J.~Fan, Y.~Feng, and L.~Xia.
\newblock A projection-based conditional dependence measure with applications
  to high-dimensional undirected graphical models.
\newblock \emph{Journal of Econometrics}, 218\penalty0 (1):\penalty0 119--139,
  2020.

\bibitem[Figalli(2018)]{figalli18}
A.~Figalli.
\newblock On the continuity of center-outward distribution and quantile
  functions.
\newblock \emph{Nonlinear Analysis}, 177:\penalty0 413--421, 2018.

\bibitem[Freund and Schapire(1996)]{freundschapire96}
Y.~Freund and R.~E. Schapire.
\newblock Experiments with a new boosting algorithm.
\newblock In \emph{Machine Learning: Proceedings of the Thirteenth
  International Conference}, pages 148--156, 1996.

\bibitem[Friedman(1991)]{friedman91}
J.~H. Friedman.
\newblock Multivariate adaptive regression splines.
\newblock \emph{Annals of Statistics}, 19\penalty0 (1):\penalty0 1--67, 1991.

\bibitem[Friedman and Rafsky(1983)]{friedmanrafsky83}
J.~H. Friedman and L.~C. Rafsky.
\newblock Graph-theoretic measures of multivariate association and prediction.
\newblock \emph{Annals of Statistics}, pages 377--391, 1983.

\bibitem[Fukumizu et~al.(2007)Fukumizu, Gretton, Sun, and
  Sch{\"o}lkopf]{fukumizuetal07}
K.~Fukumizu, A.~Gretton, X.~Sun, and B.~Sch{\"o}lkopf.
\newblock Kernel measures of conditional dependence.
\newblock In \emph{Advances in Neural Information Processing Systems},
  volume~20. Curran Associates, Inc., 2007.

\bibitem[Gamboa et~al.(2018)Gamboa, Klein, and Lagnoux]{gamboaetal18}
F.~Gamboa, T.~Klein, and A.~Lagnoux.
\newblock Sensitivity analysis based on {Cram{\'e}r--von Mises} distance.
\newblock \emph{SIAM/ASA Journal on Uncertainty Quantification}, 6\penalty0
  (2):\penalty0 522--548, 2018.

\bibitem[Gamboa et~al.(2022)Gamboa, Gremaud, Klein, and Lagnoux]{gamboaetal22}
F.~Gamboa, P.~Gremaud, T.~Klein, and A.~Lagnoux.
\newblock Global sensitivity analysis: {A} novel generation of mighty
  estimators based on rank statistics.
\newblock \emph{Bernoulli}, 28\penalty0 (4):\penalty0 2345--2374, 2022.

\bibitem[Gebelein(1941)]{gebelein41}
H.~Gebelein.
\newblock Das statistische problem der korrelation als variations-und
  eigenwertproblem und sein zusammenhang mit der ausgleichsrechnung.
\newblock \emph{Zeitschrift f{\"u}r Angewandte Mathematik und Mechanik},
  21\penalty0 (6):\penalty0 364--379, 1941.

\bibitem[George and McCulloch(1993)]{georgemcculloch93}
E.~I. George and R.~E. McCulloch.
\newblock Variable selection via {G}ibbs sampling.
\newblock \emph{Journal of the American Statistical Association}, 88\penalty0
  (423):\penalty0 881--889, 1993.

\bibitem[Ghosal and Sen(2022)]{ghosalsen22}
P.~Ghosal and B.~Sen.
\newblock Multivariate ranks and quantiles using optimal transport:
  {C}onsistency, rates and nonparametric testing.
\newblock \emph{Annals of Statistics}, 50\penalty0 (2):\penalty0 1012--1037,
  2022.

\bibitem[Gretton et~al.(2005{\natexlab{a}})Gretton, Bousquet, Smola, and
  Sch{\"o}lkopf]{grettonetal05}
A.~Gretton, O.~Bousquet, A.~Smola, and B.~Sch{\"o}lkopf.
\newblock Measuring statistical dependence with {H}ilbert--{S}chmidt norms.
\newblock In \emph{Proceedings of the 16th International Conference on
  Algorithmic Learning Theory}, pages 63--77. Springer, Berlin,
  2005{\natexlab{a}}.

\bibitem[Gretton et~al.(2005{\natexlab{b}})Gretton, Smola, Bousquet, Herbrich,
  Belitski, Augath, Murayama, Pauls, Sch{\"o}lkopf, and
  Logothetis]{grettonetal05a}
A.~Gretton, A.~Smola, O.~Bousquet, R.~Herbrich, A.~Belitski, M.~Augath,
  Y.~Murayama, J.~Pauls, B.~Sch{\"o}lkopf, and N.~Logothetis.
\newblock Kernel constrained covariance for dependence measurement.
\newblock In \emph{International Workshop on Artificial Intelligence and
  Statistics}, pages 112--119. PMLR, 2005{\natexlab{b}}.

\bibitem[Gretton et~al.(2007)Gretton, Fukumizu, Teo, Song, Sch{\"o}lkopf, and
  Smola]{grettonetal07}
A.~Gretton, K.~Fukumizu, C.~Teo, L.~Song, B.~Sch{\"o}lkopf, and A.~Smola.
\newblock A kernel statistical test of independence.
\newblock In \emph{Advances in Neural Information Processing Systems},
  volume~20. Curran Associates, Inc., 2007.

\bibitem[Hallin et~al.(2021)Hallin, Del~Barrio, Cuesta-Albertos, and
  Matr{\'a}n]{hallinetal21}
M.~Hallin, E.~Del~Barrio, J.~Cuesta-Albertos, and C.~Matr{\'a}n.
\newblock Distribution and quantile functions, ranks and signs in dimension
  $d$: {A} measure transportation approach.
\newblock \emph{Annals of Statistics}, 49\penalty0 (2):\penalty0 1139--1165,
  2021.

\bibitem[Han(2021)]{han21}
F.~Han.
\newblock On extensions of rank correlation coefficients to multivariate
  spaces.
\newblock \emph{Bernoulli News}, 28\penalty0 (2):\penalty0 7--11, 2021.

\bibitem[Han and Huang(2022)]{hanhuang22}
F.~Han and Z.~Huang.
\newblock {Azadkia--Chatterjee's correlation coefficient adapts to manifold
  data}.
\newblock \emph{arXiv preprint arXiv:2209.11156}, 2022.

\bibitem[Han et~al.(2017)Han, Chen, and Liu]{hanetal17}
F.~Han, S.~Chen, and H.~Liu.
\newblock Distribution-free tests of independence in high dimensions.
\newblock \emph{Biometrika}, 104\penalty0 (4):\penalty0 813--828, 2017.

\bibitem[Hastie et~al.(2009)Hastie, Tibshirani, and Friedman]{hastieetal09}
T.~Hastie, R.~Tibshirani, and J.~H. Friedman.
\newblock \emph{{The Elements of Statistical Learning: Data Mining, Inference,
  and Prediction}}.
\newblock Springer, Berlin, 2nd edition, 2009.

\bibitem[Heller and Heller(2016)]{hellerheller16}
R.~Heller and Y.~Heller.
\newblock Multivariate tests of association based on univariate tests.
\newblock In \emph{Advances in Neural Information Processing Systems},
  volume~29. Curran Associates, Inc., 2016.

\bibitem[Heller et~al.(2012)Heller, Gorfine, and Heller]{helleretal12}
R.~Heller, M.~Gorfine, and Y.~Heller.
\newblock A class of multivariate distribution-free tests of independence based
  on graphs.
\newblock \emph{Journal of Statistical Planning and Inference}, 142\penalty0
  (12):\penalty0 3097--3106, 2012.

\bibitem[Heller et~al.(2013)Heller, Heller, and Gorfine]{helleretal13}
R.~Heller, Y.~Heller, and M.~Gorfine.
\newblock A consistent multivariate test of association based on ranks of
  distances.
\newblock \emph{Biometrika}, 100\penalty0 (2):\penalty0 503--510, 2013.

\bibitem[Hirschfeld(1935)]{hirschfeld35}
H.~O. Hirschfeld.
\newblock A connection between correlation and contingency.
\newblock \emph{Mathematical Proceedings of the Cambridge Philosophical
  Society}, 31\penalty0 (4):\penalty0 520--524, 1935.

\bibitem[Ho(1998)]{ho98}
T.~K. Ho.
\newblock The random subspace method for constructing decision forests.
\newblock \emph{IEEE Transactions on Pattern Analysis and Machine
  Intelligence}, 20\penalty0 (8):\penalty0 832--844, 1998.

\bibitem[Hoeffding(1948)]{hoeffding48}
W.~Hoeffding.
\newblock A non-parametric test of independence.
\newblock \emph{Annals of Mathematical Statististics}, 19\penalty0
  (4):\penalty0 546--557, 1948.

\bibitem[Huang(2010)]{huang10}
T.-M. Huang.
\newblock Testing conditional independence using maximal nonlinear conditional
  correlation.
\newblock \emph{Annals of Statistics}, 38\penalty0 (4):\penalty0 2047--2091,
  2010.

\bibitem[Huang et~al.(2020)Huang, Deb, and Sen]{huangetal20}
Z.~Huang, N.~Deb, and B.~Sen.
\newblock Kernel partial correlation coefficient --- a measure of conditional
  dependence.
\newblock \emph{arXiv preprint arXiv:2012.14804}, 2020.

\bibitem[Huang et~al.(2022)Huang, Deb, and Sen]{huangetal22}
Z.~Huang, N.~Deb, and B.~Sen.
\newblock \emph{KPC: Kernel Partial Correlation Coefficient}, 2022.
\newblock URL \url{https://cran.r-project.org/web/packages/KPC}.

\bibitem[Joe(1989)]{joe89}
H.~Joe.
\newblock Relative entropy measures of multivariate dependence.
\newblock \emph{Journal of the American Statistical Association}, 84\penalty0
  (405):\penalty0 157--164, 1989.

\bibitem[Josse and Holmes(2016)]{josseholmes16}
J.~Josse and S.~Holmes.
\newblock Measuring multivariate association and beyond.
\newblock \emph{Statistics Surveys}, 10:\penalty0 132, 2016.

\bibitem[Ke and Yin(2019)]{keyin19}
C.~Ke and X.~Yin.
\newblock Expected conditional characteristic function-based measures for
  testing independence.
\newblock \emph{Journal of the American Statistical Association}, 115\penalty0
  (530):\penalty0 985--996, 2019.

\bibitem[Kim et~al.(2020)Kim, Balakrishnan, and Wasserman]{kimetal20}
I.~Kim, S.~Balakrishnan, and L.~Wasserman.
\newblock Robust multivariate nonparametric tests via projection averaging.
\newblock \emph{Annals of Statistics}, 48\penalty0 (6):\penalty0 3417--3441,
  2020.

\bibitem[Kim et~al.(2021)Kim, Neykov, Balakrishnan, and Wasserman]{kimetal21}
I.~Kim, M.~Neykov, S.~Balakrishnan, and L.~Wasserman.
\newblock Local permutation tests for conditional independence.
\newblock \emph{arXiv preprint arXiv:2112.11666}, 2021.

\bibitem[Kong et~al.(2019)Kong, Xia, and Zhong]{kongetal19}
E.~Kong, Y.~Xia, and W.~Zhong.
\newblock Composite coefficient of determination and its application in
  ultrahigh dimensional variable screening.
\newblock \emph{Journal of the American Statistical Association}, 114\penalty0
  (528):\penalty0 1740--1751, 2019.

\bibitem[Kraskov et~al.(2004)Kraskov, St{\"o}gbauer, and
  Grassberger]{kraskovetal04}
A.~Kraskov, H.~St{\"o}gbauer, and P.~Grassberger.
\newblock Estimating mutual information.
\newblock \emph{Physical Review E}, 69\penalty0 (6):\penalty0 066138, 2004.

\bibitem[Lin and Han(2022+)]{linhan21}
Z.~Lin and F.~Han.
\newblock On boosting the power of {C}hatterjee's rank correlation.
\newblock \emph{Biometrika}, 2022+.
\newblock Forthcoming.

\bibitem[Lin and Han(2022)]{linhan22}
Z.~Lin and F.~Han.
\newblock Limit theorems of {C}hatterjee's rank correlation.
\newblock \emph{arXiv preprint arXiv:2204.08031}, 2022.

\bibitem[Linfoot(1957)]{linfoot57}
E.~H. Linfoot.
\newblock An informational measure of correlation.
\newblock \emph{Information and Control}, 1\penalty0 (1):\penalty0 85--89,
  1957.

\bibitem[Linton and Gozalo(1997)]{lintongozalo97}
O.~Linton and P.~Gozalo.
\newblock Conditional independence restrictions: testing and estimation.
\newblock \emph{Cowles Foundation Discussion Paper}, 1140, 1997.

\bibitem[Lopez-Paz et~al.(2013)Lopez-Paz, Hennig, and
  Sch{\"o}lkopf]{lopez-pazetal13}
D.~Lopez-Paz, P.~Hennig, and B.~Sch{\"o}lkopf.
\newblock The randomized dependence coefficient.
\newblock In \emph{Advances in Neural Information Processing Systems},
  volume~26. Curran Associates, Inc., 2013.

\bibitem[Lyons(2013)]{lyons13}
R.~Lyons.
\newblock Distance covariance in metric spaces.
\newblock \emph{Annals of Probability}, 41\penalty0 (5):\penalty0 3284--3305,
  2013.

\bibitem[Manole et~al.(2021)Manole, Balakrishnan, Niles-Weed, and
  Wasserman]{manoleetal21}
T.~Manole, S.~Balakrishnan, J.~Niles-Weed, and L.~Wasserman.
\newblock Plugin estimation of smooth optimal transport maps.
\newblock \emph{arXiv preprint arXiv:2107.12364}, 2021.

\bibitem[Mantel and Haenszel(1959)]{mantelhaenszel59}
N.~Mantel and W.~Haenszel.
\newblock Statistical aspects of the analysis of data from retrospective
  studies of disease.
\newblock \emph{Journal of the National Cancer Institute}, 22\penalty0
  (4):\penalty0 719--748, 1959.

\bibitem[McCann(1995)]{mccann95}
R.~J. McCann.
\newblock Existence and uniqueness of monotone measure-preserving maps.
\newblock \emph{Duke Mathematical Journal}, 80\penalty0 (2):\penalty0 309--323,
  1995.

\bibitem[Miller(2002)]{miller02}
A.~Miller.
\newblock \emph{{Subset Selection in Regression}}.
\newblock Chapman and Hall, 2002.

\bibitem[Mordant and Segers(2022)]{mordantsegers22}
G.~Mordant and J.~Segers.
\newblock Measuring dependence between random vectors via optimal transport.
\newblock \emph{Journal of Multivariate Analysis}, 189:\penalty0 104912, 2022.

\bibitem[Nandy et~al.(2016)Nandy, Weihs, and Drton]{nandyetal16}
P.~Nandy, L.~Weihs, and M.~Drton.
\newblock Large-sample theory for the {B}ergsma--{D}assios sign covariance.
\newblock \emph{Electronic Journal of Statistics}, 10\penalty0 (2):\penalty0
  2287--2311, 2016.

\bibitem[Neykov et~al.(2021)Neykov, Balakrishnan, and Wasserman]{neykovetal21}
M.~Neykov, S.~Balakrishnan, and L.~Wasserman.
\newblock Minimax optimal conditional independence testing.
\newblock \emph{Annals of Statistics}, 49\penalty0 (4):\penalty0 2151--2177,
  2021.

\bibitem[Patra et~al.(2016)Patra, Sen, and Sz{\'e}kely]{patraetal16}
R.~K. Patra, B.~Sen, and G.~J. Sz{\'e}kely.
\newblock On a nonparametric notion of residual and its applications.
\newblock \emph{Statistics \& Probability Letters}, 109:\penalty0 208--213,
  2016.

\bibitem[Pfister et~al.(2018)Pfister, B{\"u}hlmann, Sch{\"o}lkopf, and
  Peters]{pfisteretal18}
N.~Pfister, P.~B{\"u}hlmann, B.~Sch{\"o}lkopf, and J.~Peters.
\newblock Kernel-based tests for joint independence.
\newblock \emph{Journal of the Royal Statistical Society: Series B (Statistical
  Methodology)}, 80\penalty0 (1):\penalty0 5--31, 2018.

\bibitem[P{\'o}czos and Schneider(2012)]{poczosschneider12}
B.~P{\'o}czos and J.~Schneider.
\newblock Nonparametric estimation of conditional information and divergences.
\newblock In \emph{Artificial Intelligence and Statistics}, pages 914--923.
  PMLR, 2012.

\bibitem[Puri et~al.(1970)Puri, Sen, and Gokhale]{purietal70}
M.~Puri, P.~Sen, and D.~Gokhale.
\newblock On a class of rank order tests for independence in multivariate
  distributions.
\newblock \emph{Sankhy\=a, Series A}, 32\penalty0 (3):\penalty0 271--298, 1970.

\bibitem[Puri and Sen(1971)]{purisen71}
M.~L. Puri and P.~K. Sen.
\newblock \emph{Nonparametric methods in multivariate analysis}.
\newblock Wiley, New York, 1971.

\bibitem[Rao and Srivastava(1994)]{raosrivastava94}
B.~Rao and S.~Srivastava.
\newblock An elementary proof of the {B}orel isomorphism theorem.
\newblock \emph{Real Analysis Exchange}, 20\penalty0 (1):\penalty0 347--349,
  1994.

\bibitem[Ravikumar et~al.(2009)Ravikumar, Lafferty, Liu, and
  Wasserman]{ravikumaretal09}
P.~Ravikumar, J.~Lafferty, H.~Liu, and L.~Wasserman.
\newblock Sparse additive models.
\newblock \emph{Journal of the Royal Statistical Society: Series B (Statistical
  Methodology)}, 71\penalty0 (5):\penalty0 1009--1030, 2009.

\bibitem[R{\'e}nyi(1959)]{renyi59}
A.~R{\'e}nyi.
\newblock On measures of dependence.
\newblock \emph{Acta Mathematica Hungarica}, 10\penalty0 (3-4):\penalty0
  441--451, 1959.

\bibitem[Reshef et~al.(2011)Reshef, Reshef, Finucane, Grossman, McVean,
  Turnbaugh, Lander, Mitzenmacher, and Sabeti]{reshefetal11}
D.~N. Reshef, Y.~A. Reshef, H.~K. Finucane, S.~R. Grossman, G.~McVean, P.~J.
  Turnbaugh, E.~S. Lander, M.~Mitzenmacher, and P.~C. Sabeti.
\newblock Detecting novel associations in large data sets.
\newblock \emph{Science}, 334\penalty0 (6062):\penalty0 1518--1524, 2011.

\bibitem[Romano(1988)]{romano88}
J.~P. Romano.
\newblock A bootstrap revival of some nonparametric distance tests.
\newblock \emph{Journal of the American Statistical Association}, 83\penalty0
  (403):\penalty0 698--708, 1988.

\bibitem[Rosenblatt(1975)]{rosenblatt75}
M.~Rosenblatt.
\newblock A quadratic measure of deviation of two-dimensional density estimates
  and a test of independence.
\newblock \emph{Annals of Statistics}, pages 1--14, 1975.

\bibitem[Runge(2018)]{runge18}
J.~Runge.
\newblock Conditional independence testing based on a nearest-neighbor
  estimator of conditional mutual information.
\newblock In \emph{International Conference on Artificial Intelligence and
  Statistics}, pages 938--947. PMLR, 2018.

\bibitem[Sadeghi(2022)]{sadeghi22}
B.~Sadeghi.
\newblock {Chatterjee Correlation Coefficient: a robust alternative for classic
  correlation methods in geochemical studies-(including “TripleCpy” Python
  package)}.
\newblock \emph{Ore Geology Reviews}, page 104954, 2022.

\bibitem[Schweizer and Wolff(1981)]{schweizerwolff81}
B.~Schweizer and E.~F. Wolff.
\newblock On nonparametric measures of dependence for random variables.
\newblock \emph{Annals of Statistics}, 9\penalty0 (4):\penalty0 879--885, 1981.

\bibitem[Sen and Sen(2014)]{sensen14}
A.~Sen and B.~Sen.
\newblock Testing independence and goodness-of-fit in linear models.
\newblock \emph{Biometrika}, 101\penalty0 (4):\penalty0 927--942, 2014.

\bibitem[Sen et~al.(2017)Sen, Suresh, Shanmugam, Dimakis, and
  Shakkottai]{senetal17}
R.~Sen, A.~T. Suresh, K.~Shanmugam, A.~G. Dimakis, and S.~Shakkottai.
\newblock Model-powered conditional independence test.
\newblock In \emph{Advances in Neural Information Processing Systems},
  volume~30. Curran Associates, Inc., 2017.

\bibitem[Seth and Pr{\'\i}ncipe(2012)]{sethprincipe12}
S.~Seth and J.~C. Pr{\'\i}ncipe.
\newblock Conditional association.
\newblock \emph{Neural Computation}, 24\penalty0 (7):\penalty0 1882--1905,
  2012.

\bibitem[Shah and Peters(2020)]{shahpeters20}
R.~D. Shah and J.~Peters.
\newblock The hardness of conditional independence testing and the generalised
  covariance measure.
\newblock \emph{Annals of Statistics}, 48\penalty0 (3):\penalty0 1514--1538,
  2020.

\bibitem[Shi et~al.(2021{\natexlab{a}})Shi, Drton, Hallin, and Han]{shietal21a}
H.~Shi, M.~Drton, M.~Hallin, and F.~Han.
\newblock Center-outward sign-and rank-based quadrant, {S}pearman, and
  {K}endall tests for multivariate independence.
\newblock \emph{arXiv preprint arXiv:2111.15567}, 2021{\natexlab{a}}.

\bibitem[Shi et~al.(2021{\natexlab{b}})Shi, Drton, and Han]{shietal21}
H.~Shi, M.~Drton, and F.~Han.
\newblock {On Azadkia--Chatterjee's conditional dependence coefficient}.
\newblock \emph{arXiv preprint arXiv:2108.06827}, 2021{\natexlab{b}}.

\bibitem[Shi et~al.(2022{\natexlab{a}})Shi, Drton, and Han]{shietal22}
H.~Shi, M.~Drton, and F.~Han.
\newblock {On the power of {C}hatterjee's rank correlation}.
\newblock \emph{Biometrika}, 109\penalty0 (2):\penalty0 317--333,
  2022{\natexlab{a}}.

\bibitem[Shi et~al.(2022{\natexlab{b}})Shi, Drton, and Han]{shietal22a}
H.~Shi, M.~Drton, and F.~Han.
\newblock Distribution-free consistent independence tests via center-outward
  ranks and signs.
\newblock \emph{Journal of the American Statistical Association}, 117\penalty0
  (537):\penalty0 395--410, 2022{\natexlab{b}}.

\bibitem[Shi et~al.(2022{\natexlab{c}})Shi, Hallin, Drton, and Han]{shietal22b}
H.~Shi, M.~Hallin, M.~Drton, and F.~Han.
\newblock On universally consistent and fully distribution-free rank tests of
  vector independence.
\newblock \emph{Annals of Statistics}, 50\penalty0 (4):\penalty0 1933--1959,
  2022{\natexlab{c}}.

\bibitem[Sklar(1959)]{sklar59}
M.~Sklar.
\newblock Fonctions de r\'epartition \`a $n$ dimensions et leurs marges.
\newblock \emph{Publ. Inst. Statist. Univ. Paris}, 8:\penalty0 229--231, 1959.

\bibitem[Song(2009)]{song09}
K.~Song.
\newblock Testing conditional independence via {R}osenblatt transforms.
\newblock \emph{Annals of Statistics}, 37\penalty0 (6B):\penalty0 4011--4045,
  2009.

\bibitem[Srivastava(1998)]{srivastava98}
S.~M. Srivastava.
\newblock \emph{{A Course on Borel Sets}}.
\newblock Springer-Verlag, New York, 1998.

\bibitem[Strobl et~al.(2019)Strobl, Zhang, and Visweswaran]{strobletal19}
E.~V. Strobl, K.~Zhang, and S.~Visweswaran.
\newblock Approximate kernel-based conditional independence tests for fast
  non-parametric causal discovery.
\newblock \emph{Journal of Causal Inference}, 7\penalty0 (1), 2019.

\bibitem[Su and White(2007)]{suwhite07}
L.~Su and H.~White.
\newblock A consistent characteristic function-based test for conditional
  independence.
\newblock \emph{Journal of Econometrics}, 141\penalty0 (2):\penalty0 807--834,
  2007.

\bibitem[Su and White(2008)]{suwhite08}
L.~Su and H.~White.
\newblock A nonparametric {H}ellinger metric test for conditional independence.
\newblock \emph{Econometric Theory}, 24\penalty0 (4):\penalty0 829--864, 2008.

\bibitem[Su and White(2014)]{suwhite14}
L.~Su and H.~White.
\newblock Testing conditional independence via empirical likelihood.
\newblock \emph{Journal of Econometrics}, 182\penalty0 (1):\penalty0 27--44,
  2014.

\bibitem[Sz{\'e}kely and Rizzo(2009)]{szekelyrizzo09}
G.~J. Sz{\'e}kely and M.~L. Rizzo.
\newblock Brownian distance covariance.
\newblock \emph{Annals of Applied Statistics}, 3\penalty0 (4):\penalty0
  1236--1265, 2009.

\bibitem[Sz{\'e}kely and Rizzo(2014)]{szekelyrizzo14}
G.~J. Sz{\'e}kely and M.~L. Rizzo.
\newblock Partial distance correlation with methods for dissimilarities.
\newblock \emph{Annals of Statistics}, 42\penalty0 (6):\penalty0 2382--2412,
  2014.

\bibitem[Sz{\'e}kely et~al.(2007)Sz{\'e}kely, Rizzo, and
  Bakirov]{szekelyetal07}
G.~J. Sz{\'e}kely, M.~L. Rizzo, and N.~K. Bakirov.
\newblock Measuring and testing dependence by correlation of distances.
\newblock \emph{Annals of Statistics}, 35\penalty0 (6):\penalty0 2769--2794,
  2007.

\bibitem[Tibshirani(1996)]{tibshirani96}
R.~Tibshirani.
\newblock Regression shrinkage and selection via the lasso.
\newblock \emph{Journal of the Royal Statistical Society: Series B
  (Methodological)}, 58\penalty0 (1):\penalty0 267--288, 1996.

\bibitem[Veraverbeke et~al.(2011)Veraverbeke, Omelka, and
  Gijbels]{veraverbekeetal11}
N.~Veraverbeke, M.~Omelka, and I.~Gijbels.
\newblock Estimation of a conditional copula and association measures.
\newblock \emph{Scandinavian Journal of Statistics}, 38\penalty0 (4):\penalty0
  766--780, 2011.

\bibitem[Vergara and Est{\'e}vez(2014)]{vergaraestevez14}
J.~R. Vergara and P.~A. Est{\'e}vez.
\newblock A review of feature selection methods based on mutual information.
\newblock \emph{Neural Computing and Applications}, 24\penalty0 (1):\penalty0
  175--186, 2014.

\bibitem[Wang et~al.(2015)Wang, Pan, Hu, Tian, and Zhang]{wangetal15}
X.~Wang, W.~Pan, W.~Hu, Y.~Tian, and H.~Zhang.
\newblock Conditional distance correlation.
\newblock \emph{Journal of the American Statistical Association}, 110\penalty0
  (512):\penalty0 1726--1734, 2015.

\bibitem[Wang et~al.(2017)Wang, Jiang, and Liu]{wangetal17}
X.~Wang, B.~Jiang, and J.~S. Liu.
\newblock Generalized {R}-squared for detecting dependence.
\newblock \emph{Biometrika}, 104\penalty0 (1):\penalty0 129--139, 2017.

\bibitem[Weihs et~al.(2016)Weihs, Drton, and Leung]{weihsetal16}
L.~Weihs, M.~Drton, and D.~Leung.
\newblock Efficient computation of the {B}ergsma--{D}assios sign covariance.
\newblock \emph{Computational Statistics}, 31\penalty0 (1):\penalty0 315--328,
  2016.

\bibitem[Weihs et~al.(2018)Weihs, Drton, and Meinshausen]{weihsetal18}
L.~Weihs, M.~Drton, and N.~Meinshausen.
\newblock Symmetric rank covariances: a generalized framework for nonparametric
  measures of dependence.
\newblock \emph{Biometrika}, 105\penalty0 (3):\penalty0 547--562, 2018.

\bibitem[Yanagimoto(1970)]{yanagimoto70}
T.~Yanagimoto.
\newblock On measures of association and a related problem.
\newblock \emph{Annals of the Institute of Statistical Mathematics},
  22\penalty0 (1):\penalty0 57--63, 1970.

\bibitem[Yuan and Lin(2006)]{yuanlin06}
M.~Yuan and Y.~Lin.
\newblock Model selection and estimation in regression with grouped variables.
\newblock \emph{Journal of the Royal Statistical Society: Series B (Statistical
  Methodology)}, 68\penalty0 (1):\penalty0 49--67, 2006.

\bibitem[Zhang(2019)]{zhang19}
K.~Zhang.
\newblock {BET} on independence.
\newblock \emph{Journal of the American Statistical Association}, 114\penalty0
  (528):\penalty0 1620--1637, 2019.

\bibitem[Zhang et~al.(2012)Zhang, Peters, Janzing, and
  Sch{\"o}lkopf]{zhangetal12}
K.~Zhang, J.~Peters, D.~Janzing, and B.~Sch{\"o}lkopf.
\newblock Kernel-based conditional independence test and application in causal
  discovery.
\newblock \emph{arXiv preprint arXiv:1202.3775}, 2012.

\bibitem[Zhang(2022)]{zhang22}
Q.~Zhang.
\newblock On the asymptotic distribution of the symmetrized {C}hatterjee's
  correlation coefficient.
\newblock \emph{arXiv preprint arXiv:2205.01769}, 2022.

\bibitem[Zhang et~al.(2018)Zhang, Filippi, Gretton, and
  Sejdinovic]{zhangetal18}
Q.~Zhang, S.~Filippi, A.~Gretton, and D.~Sejdinovic.
\newblock Large-scale kernel methods for independence testing.
\newblock \emph{Statistics and Computing}, 28\penalty0 (1):\penalty0 113--130,
  2018.

\bibitem[Zhu et~al.(2017)Zhu, Xu, Li, and Zhong]{zhuetal17}
L.~Zhu, K.~Xu, R.~Li, and W.~Zhong.
\newblock Projection correlation between two random vectors.
\newblock \emph{Biometrika}, 104\penalty0 (4):\penalty0 829--843, 2017.

\bibitem[Zou(2006)]{zou06}
H.~Zou.
\newblock The adaptive lasso and its oracle properties.
\newblock \emph{Journal of the American Statistical Association}, 101\penalty0
  (476):\penalty0 1418--1429, 2006.

\bibitem[Zou and Hastie(2005)]{zouhastie05}
H.~Zou and T.~Hastie.
\newblock Regularization and variable selection via the elastic net.
\newblock \emph{Journal of the Royal Statistical Society: Series B (Statistical
  Methodology)}, 67\penalty0 (2):\penalty0 301--320, 2005.

\end{thebibliography}

\end{document}